\documentclass[runningheads]{llncs}

\usepackage[T1]{fontenc}
\usepackage{lmodern}
\usepackage{amsmath,amssymb,mathtools}
\usepackage[ruled,vlined,linesnumbered]{algorithm2e}
\usepackage{booktabs}
\usepackage{xspace}
\usepackage{url}
\usepackage{microtype}
\usepackage{xcolor}

\usepackage{orcidlink}
\providecommand{\orcidID}[1]{}
\renewcommand{\orcidID}[1]{\orcidlink{#1}}

\newcommand{\SA}{\mathsf{SA}}
\newcommand{\ISA}{\mathsf{ISA}}
\newcommand{\LCP}{\mathsf{LCP}}

\newcommand{\LCE}{\mathsf{LCE}}
\newcommand{\range}{\mathsf{range}}

\newcommand{\MRep}{\mathsf{MRep}}
\newcommand{\occ}{\mathsf{occ}}
\newcommand{\eL}{e_{\mathrm{L}}}
\newcommand{\eR}{e_{\mathrm{R}}}
\newcommand{\Substr}{\mathsf{Substr}}

\newcommand{\LChar}{\mathsf{LChar}}
\newcommand{\RChar}{\mathsf{RChar}}
\newcommand{\DAWG}{\mathsf{DAWG}}

\newcommand{\GNP}{\mathsf{GNP}}
\newcommand{\sGNP}{s_{\GNP}}
\newcommand{\tGNP}{t_{\GNP}}
\newcommand{\eps}{\varepsilon}

\newcommand{\leftcl}{\overleftarrow}
\newcommand{\rightcl}{\overrightarrow}
\newcommand{\bothcl}{\overleftrightarrow}
\newcommand{\rank}{\mathsf{rank}}
\newcommand{\leftpos}{\mathsf{left}}
\newcommand{\rightpos}{\mathsf{right}}

\newcommand{\str}{\mathsf{str}}

\title{Output-Sensitive Construction of CDAWGs from BWT-Runs}
\titlerunning{Output-Sensitive Construction of CDAWGs from BWT-Runs}

\author{Yuta Tsuruzono\inst{1} \and Hiroki Arimura\inst{1}\orcidID{0000-0002-2701-0271} \and Shunsuke Inenaga\inst{2}\orcidID{0000-0002-1833-010X}}

\authorrunning{Tsuruzono, Arimura, Inenaga}

\institute{Graduate School of IST, Hokkaido University, Sapporo, Japan\\
\email{arim@ist.hokudai.ac.jp}
\and
Department of Informatics, Kyushu University, Fukuoka, Japan\\
\email{inenaga.shunsuke.380@m.kyushu-u.ac.jp}}

\begin{document}
\maketitle

\begin{abstract}
  The \emph{compact directed acyclic word graph} (\emph{CDAWG}) of a string can be viewed in two equivalent ways: as the edge-compacted DAWG of the string, and as the DAG obtained from the suffix tree by merging the nodes whose subtrees are isomorphic.  By exploiting these two views in opposite directions,
  we show how to build, for the (reversed) input string of length $n$, the CDAWG with $\eL$ edges in $O(\eL\log n\log(n/r))$ time with $O(r\log(n/r)+\eL)$ words of working space,
  provided that the fully functional compressed suffix tree of Gagie, Navarro, and Prezza of size $O(r\log(n/r))$ is available.
  Here, $r$ denotes the number of BWT-runs of the input string.
%  This gives the first output-sensitive CDAWG construction using BWT-run-bounded working space apart from the output.
%\keywords{CDAWG \and DAWG \and Weiner links \and BWT-runs \and compressed suffix tree}
\end{abstract}

\section{Introduction}

The \emph{directed acyclic word graph} (\emph{DAWG}) of a string is the smallest automaton recognizing all substrings of the string~\cite{BlumerEtAl1985}.  Its edge-compacted version is the compact directed acyclic word graph (CDAWG)~\cite{BlumerEtAl1987}.  Equivalently, the CDAWG can be obtained from the suffix tree~\cite{Weiner1973} by merging the nodes whose subtrees are isomorphic.  This second view is particularly useful algorithmically: a CDAWG node can be identified by a suffix-tree subtree-equivalence class, while its represented strings are maximal repeats~\cite{NarisawaEtAl2017}.  Thus their size can be much smaller than the text length for highly repetitive strings.

A second classical ingredient is the Weiner-link structure of the suffix tree~\cite{Weiner1973}.  A Weiner link prepends a character to the represented string.
Under reversal, such left extensions become ordinary edges in the DAWG of the reversed string.
% DELETED/REPLACED to avoid repeating the hard/soft = primary/secondary explanation here; the formal statement is in Preliminaries.
%Our construction is based on this correspondence.  We traverse Weiner links on the forward suffix-tree side, but we interpret the added edge representations as edges of the reversed-string DAWG; the rule deciding which edges are to be compacted is supplied by the suffix-tree subtree-isomorphism criterion on the forward string.

Another well-studied measure of repetitiveness is the number $r$ of equal-letter runs in the Burrows-Wheeler transform (BWT).  The r-index and the fully functional compressed suffix tree of Gagie, Navarro, and Prezza~\cite{GagieNavarroPrezza2020} provide suffix-array and suffix-tree functionality in BWT-run bounded space.  Nishimoto and Tabei's r-enum~\cite{NishimotoTabei2021} showed that several characteristic substrings, including maximal repeats, can be enumerated in $O(r)$ words of working space by traversing all right-maximal repeats through Weiner links.
Kimura and I~\cite{KimuraI2026} subsequently reduced the traversal time to $O(n)$ by using the move data structure for constant-time LF-mapping.
%Their bound charges range-distinct outputs and LF-mapping operations to all Weiner links encountered during the traversal.  Consequently, it remains linear in the text length even when the number of maximal repeats, or the size of the target CDAWG, is much smaller.  In contrast, our objective is to avoid enumerating all right-maximal repeats and to charge the traversal only to left extensions of maximal repeats, skipping unary regions that do not create CDAWG nodes.

In this paper, we present an output-sensitive construction of CDAWGs from BWT-run bounded primitives.
Since the maximal repeats are the nodes of the CDAWG,
our algorithm can be seen as a generalization and extension of the aforementioned r-enum algorithms.
Namely, our algorithm starts from the empty string and traverses Weiner links, in the spirit of r-enum.  
Hard and soft Weiner links play different conceptual roles.  A hard Weiner link leads to a right-maximal locus and is therefore the same kind of step that r-enum uses for maximal-repeat discovery; in our setting, it also represents a primary DAWG edge under string reversal.  A soft Weiner link, on the other hand, is not needed merely to enumerate maximal repeats, but it is needed for CDAWG construction because it represents a secondary DAWG edge under string reversal.  Thus processing soft Weiner links is one essential difference from r-enum.
Computationally, both cases are then handled by the same closure routine: for each left extension $aw$ of a maximal repeat, we compute the right closure $q=\rightcl{aw}$ and then $h=\leftcl{q}=\bothcl{aw}$ by a doubling-then-binary-search procedure on suffix-array intervals.  The algorithm uses $\mathsf{key}(h)$ as the target key, and equal target keys are merged according to the suffix-tree subtree-equivalence criterion of Narisawa et al.~\cite{NarisawaEtAl2017}.
% DELETED/REPLACED: the former text described only the uniform closure computation and did not make explicit the conceptual roles of hard and soft Weiner links or the difference from r-enum.
% DELETED/REPLACED: the former text separately described skipping unary Weiner-link paths, testing right-maximality, testing left-maximality, and closure; these are the same closure computation described below.

%Our construction is left-oriented:
Let $T=\$S\#$ be the input string, and let $r$ be the number of BWT-runs of $T$.  Let $\eL$ be the number of left-extensions of maximal repeats in $T$.  We show that the CDAWG topology of $S^R\$$ can be constructed in $O(\eL\log n\log(n/r))$ time using $O(r\log(n/r)+\eL)$ words of working space.  Here $\eL$ is measured on the forward working string $T$, whereas the output CDAWG is for the reversed string $S^R\$$; equivalently, $\eL(T)=\eR(T^R)$, and the final marker deletion affects only a single edge from the node for $\eps$.  This bound is obtained by combining our output-sensitive Weiner-link traversal with the fully functional compressed suffix tree of Gagie, Navarro, and Prezza~\cite{GagieNavarroPrezza2020}.
We also compute the metadata needed to represent edge labels without expanding them explicitly within the claimed bounds, following the CDAWG-grammar framework of Belazzougui and Cunial~\cite{BelazzouguiCunial2017CPM,BelazzouguiCunial2017SPIRE}.

\section{Preliminaries}

%\subsection{Strings}
\noindent \textbf{Strings.}
Let $\Sigma$ be an alphabet.  An element of $\Sigma$ is called a character, and an element of $\Sigma^*$ is called a string.  For a string $S$, let $|S|$ denote its length.  The empty string is denoted by $\eps$.  For any $1\leq i\leq |S|$, $S[i]$ denotes the $i$-th character of $S$.  For any $1\leq i\leq j\leq |S|$, $S[i..j]$ denotes the substring of $S$ that begins at position $i$ and ends at position $j$.  For convenience, let $S[i..j]=\eps$ if $i>j$.  Let $S^R=S[|S|]\cdots S[1]$ denote the reversed string of $S$.

A string $w$ occurs in $S$ iff $w\in\Substr(S)$.  Let $\occ_S(w)$ denote the number of occurrences of $w$ in $S$.  We use two special characters $\$$ and $\#$ that do not belong to the input alphabet and occur only when explicitly introduced.

%\subsection{Maximal repeats, DAWGs, and CDAWGs}
\noindent \textbf{Maximal repeats and CDAWGs.}
For a substring $w$ of $S$, let $\LChar_S(w)$ and $\RChar_S(w)$ be the sets of characters that occur immediately to the left and to the right of an occurrence of $w$, respectively.  A substring $w$ is left-maximal if $|\LChar_S(w)|\ge 2$, right-maximal if $|\RChar_S(w)|\ge 2$, and a maximal repeat if it is both left-maximal and right-maximal.  Let $\MRep(S)$ be the set of maximal repeats of $S$.  We define
\[
  \eL(S)=\sum_{w\in\MRep(S)} |\LChar_S(w)|,
  \qquad
  \eR(S)=\sum_{w\in\MRep(S)} |\RChar_S(w)|.
\]
Thus $\eL(S)$ and $\eR(S)$ are the numbers of left- and right-extensions of maximal repeats in $S$, respectively.
Since $\MRep(S^R)=\{w^R\mid w\in\MRep(S)\}$ holds for any string $S$, the CDAWG of $S^R$ naturally represents left-extensions of maximal repeats of $S$; in particular, $\eR(S^R)=\eL(S)$.
For a substring $x$ of $S$, the right-closure $\rightcl{x}$ is the shortest right-maximal string $y$ such that $x$ is a prefix of $y$ and $\occ_S(x)=\occ_S(y)$.  The left-closure $\leftcl{x}$ is defined symmetrically.  We write $\bothcl{x}=\leftcl{\rightcl{x}}$.
%Note that $\MRep(S) \cup \{S\} = \{\bothcl{x} : x \in \Substr(S)\}$.
For substrings with at least two occurrences, applying both closures yields maximal repeats.

The DAWG of $S$, denoted $\DAWG(S)$, is the smallest automaton recognizing the substrings of $S$~\cite{BlumerEtAl1985}.  The CDAWG of $S$, denoted $\mathrm{CDAWG}(S)$, is obtained by compacting non-branching paths of the DAWG~\cite{BlumerEtAl1987}.  We use the standard correspondence that CDAWG nodes, except for the sink, are represented by maximal repeats.  For a CDAWG node $v$, let $\str(v)$ denote the longest string represented by $v$.  The outgoing edges from the node corresponding to a maximal repeat $w$ are in one-to-one correspondence with the right-extensions of $w$; consequently, the number of CDAWG edges leaving maximal-repeat nodes is $\eR(S)$.
We also use the suffix-tree characterization of CDAWGs.  If one considers the suffix tree of $S$, then merging loci whose rooted subtrees are isomorphic yields the CDAWG of $S$~\cite{BlumerEtAl1987}.  Thus the node-identification rule for CDAWG nodes can be expressed without referring to the DAWG edges themselves: two nodes are identified exactly when their suffix-tree subtrees have the same branching structure and edge labels.
Following the work by Narisawa et al.~\cite{NarisawaEtAl2017}, our algorithm uses this characterization as the rule for deciding which intermediate endpoints are to be merged.

A Weiner link~\cite{Weiner1973} in the suffix tree prepends a character: it maps the locus of $w$ to the locus of $aw$, when such a locus exists.  Under string reversal, the left extension $aw$ of $w$ corresponds to the right extension $w^Ra$ of $w^R$.  Consequently, hard and soft Weiner links on the suffix-tree side of $S$ correspond to primary and secondary edges, respectively, of $\DAWG(S^R)$ in the terminology of Blumer et al.~\cite{BlumerEtAl1985,BlumerEtAl1987}.  After compacting the DAWG edges of $S^R$, one obtains $\mathrm{CDAWG}(S^R)$.  This is the directionality exploited in this paper: we traverse Weiner links on the forward string, but the edge representations added by the algorithm are interpreted as DAWG/CDAWG edges of the reversed string.

%\subsection{Suffix-array primitives}
\noindent \textbf{Suffix array primitives in BWT-run bounds.}
Let $S$ be a string of length $n$.  The suffix array, inverse suffix array, and LCP array of $S$ are denoted by $\SA$, $\ISA$, and $\LCP$, respectively.  For two text positions $i,j$, $\LCE(i,j)$ denotes the length of the longest common prefix of suffixes $S[i..n]$ and $S[j..n]$.  For a substring $w$ of $S$, let $\range(w)=[L_w,R_w]$ denote its suffix-array interval.
We assume access to a fully functional compressed suffix-tree interface that supports exactly the primitives used by the algorithm: $\SA$, $\ISA$, $\LCP$, $\LCE$, LF-mapping, BWT range-distinct queries, and the suffix-tree navigation needed for the subtree-equivalence test.  For a string $S$, let $\sGNP(S)$ be the number of words used by this black-box representation, and let $\tGNP(S)$ be an upper bound on the time for any one of these required primitives.  With the fully functional compressed suffix tree of Gagie, Navarro, and Prezza~\cite{GagieNavarroPrezza2020}, these primitives can be supported with $\sGNP(S)=O(r\log(n/r))$ words and $\tGNP(S)=O(\log(n/r))$ time, where $n=|S|$ and $r$ is the number of runs in the BWT of $S$.

\section{Construction Framework}

Let $S$ be the input string and let $T=\$S\#$ be the working string, where $\$$ and $\#$ are fresh characters.  After constructing the left-oriented topology for $T$, we reverse all edges and delete the unique edge from the node for $\eps$ whose label begins with $\#$, obtaining $\mathrm{CDAWG}(S^R\$)$.
Our construction framework is given in Algorithm~\ref{alg:framework}.
\begin{algorithm}[H]
\caption{Left-oriented construction on $T=\$S\#$}
\label{alg:framework}
Initialize the traversal at $\eps$ and maintain a set of discovered maximal repeats\;
\ForEach{discovered maximal repeat $w$}{
  Enumerate all left-extension triples $(a,p_a,q_a)$ of $w$ by a BWT range-distinct query\;
  \ForEach{triple $(a,p_a,q_a)$}{
    Compute $\range(aw)$ by LF-mapping\;
    \If{$\range(aw)$ is a singleton}{
      Add an edge representation with source $\mathsf{key}(w)$ and target the sink\;
    }
    \Else{
      Compute $q=\rightcl{aw}$ and $h=\leftcl{q}=\bothcl{aw}$\;
      Put $m=|q|$ and $d=|h|-m$\;
      Add an edge representation from $\mathsf{key}(w)$ to the canonical key $\mathsf{key}(h)$, storing $a$, $m$, and $d$\;
    }
  }
}
Sort all canonical keys and merge equal keys\;
Reverse the resulting edge orientation and delete the unique edge from the node for $\eps$ beginning with $\#$\;
Output $\mathrm{CDAWG}(S^R\$)$ with compact edge-label representations\;
\end{algorithm}

In the sequel, we briefly describe our algorithm.

For a substring $x$ of $T$, let $\range(x)=[L_x,R_x]$ be its suffix-array interval.  We represent a suffix-tree locus of string depth $d$ and interval $[L,R]$ by the key $\mathsf{key}([L,R],d)$.
For a maximal repeat $w\in\MRep(T)$, we write $\mathsf{key}(w)=\mathsf{key}([L_w,R_w],|w|)$.
This representation is unique for maximal repeats; in particular, $\mathsf{key}(\eps)=\mathsf{key}([1,n],0)$ represents the node for $\eps$.  The sink is not represented by an interval-depth key.

CDAWG nodes are obtained by merging suffix-tree loci whose rooted subtrees are isomorphic~\cite{BlumerEtAl1987,NarisawaEtAl2017}.
Our construction performs this merging by assigning the same canonical key to edge targets that represent the same maximal-repeat node.  Hard links are the maximal-repeat-discovery steps familiar from r-enum and give primary DAWG edges after reversal, whereas soft links are additionally needed to output secondary DAWG edges.  Algorithmically, both cases use the same rule: for a left extension $aw$ of a maximal repeat $w$, the added edge representation has source $\mathsf{key}(w)$ and has either the sink as its target, if $aw$ is unique, or the canonical key $\mathsf{key}(h)$, where $h=\bothcl{aw}$.

An edge representation whose target is a maximal-repeat node is a tuple $(\mathsf{key}(w),\mathsf{key}(h),a,m,d)$,
where $w\in\MRep(T)$, $a\in\LChar_T(w)$, $h=\bothcl{aw}$, $m=|\rightcl{aw}|$, and $d=|h|-m$.  The values $m$ and $d$ are part of the edge-label representation after reversing the orientation.  An edge representation whose target is the sink stores the information needed to represent its suffix-tail label.

All interval operations are implemented by a black-box fully functional compressed suffix tree for $T$,
which supports $\SA$, $\ISA$, $\LCP$, $\LCE$, LF-mapping, BWT range-distinct queries, and the suffix-tree navigation needed for subtree-equivalence tests.  The time and space bounds of this interface are denoted by $\tGNP(T)$ and $\sGNP(T)$, respectively.
The number of processed left extensions is $O(\eL(T))$, where $\eL(T)$ denotes the number of left-extensions of maximal repeats in $T$.  The formal computation of $\rightcl{aw}$ and $\leftcl{\rightcl{aw}}$ is given in the next section.

\section{Traversing Maximal Repeats}\label{sec:traversal}

This section gives the interval operations used in Algorithm~\ref{alg:framework}.  Throughout this section, all suffix-array intervals are taken with respect to the working string $T$ of length $n$.

\subsection{Enumerating one-step left extensions}

Let $w\in\MRep(T)$ and let $\range(w)=[L_w,R_w]$.  A BWT range-distinct query on $[L_w,R_w]$ returns the triples
\[
  (a,p_a,q_a)
\]
such that the occurrences of $w$ preceded by $a$ correspond to the BWT subrange $[p_a,q_a]$.  Applying LF-mapping gives
\[
  \range(aw)=[\mathsf{LF}(p_a),\mathsf{LF}(q_a)].
\]
Thus every left-extension of $w$ is enumerated once, and no other character is enumerated.  If this interval is a singleton, then $aw$ has a unique occurrence and the corresponding edge has the sink as its target.

Assume henceforth that $[L,R]=\range(aw)$ satisfies $L<R$.  Let
\[
  q=\rightcl{aw}
\]
be the right-closure of $aw$.  Its string depth is the common prefix length of all suffixes in the interval $[L,R]$.  For suffix-array intervals, this value is obtained from the two lexicographically extreme suffixes, and hence
% DELETED/REPLACED: the former sentence described the same right-closure operation as following a unary suffix-tree path.
\[
  |q|=\LCE(\SA[L],\SA[R]).
\]
We denote this value by $m$.

\subsection{Testing a candidate left growth}

It remains to compute $\leftcl{q}$.  For an integer $d\geq 0$, define
\[
  i_d=\SA[L]-d,\qquad j_d=\SA[R]-d.
\]
If $i_d<1$ or $j_d<1$, then $d$ is invalid.  Otherwise put
\[
  \alpha_d=\min\{\ISA[i_d],\ISA[j_d]\},\qquad
  \beta_d=\max\{\ISA[i_d],\ISA[j_d]\},
\]
and the candidate length is $m+d$.  Thus $\alpha_d$ and $\beta_d$ are the suffix-array ranks of the two leaves obtained by moving $d$ characters to the left from the leftmost and rightmost leaves of $\range(q)=[L,R]$.
% DELETED/REPLACED: the previous predicate also used the two outside LCP-boundary tests for certifying an exact suffix-array interval.  For the left-growth search, the relevant test is instead whether the two shifted boundary leaves still span exactly $R-L+1$ suffixes and share the candidate prefix.
We define the predicate $\mathsf{Valid}(L,R,m,d)$ to be true iff the following two conditions hold:
\[
  \LCE(\SA[\alpha_d],\SA[\beta_d])\geq m+d,
\]
\[
  \beta_d-\alpha_d=R-L.
\]
The first condition checks, by one LCE query, whether the two shifted boundary leaves share the candidate string of length $m+d$.  The second condition checks that these two shifted leaves still span exactly $R-L+1$ suffix-array positions, i.e., the occurrence count of $q$ is preserved.  No LCE query against the outside neighbors of $[\alpha_d,\beta_d]$ is needed in this left-growth test.

\begin{lemma}\label{lem:valid-left-growth}
For every $d\geq 0$, $\mathsf{Valid}(L,R,m,d)$ holds iff all occurrences of $q$ can be extended by the same length-$d$ string to the left.  If it holds, then the resulting string has suffix-array interval $[\alpha_d,\beta_d]$.
\end{lemma}
\begin{proof}
Suppose first that $\mathsf{Valid}(L,R,m,d)$ holds.  Let $x$ be the common prefix of length $m+d$ of the two suffixes $T[i_d..n]$ and $T[j_d..n]$.  Since suffix-array intervals are contiguous, every suffix whose rank lies between $\alpha_d$ and $\beta_d$ also has prefix $x$.  By $\beta_d-\alpha_d=R-L$, this gives exactly $R-L+1$ suffixes.  Removing the first $d$ characters from these occurrences of $x$ yields $R-L+1$ distinct occurrences of $q$.  Since $\range(q)=[L,R]$ has exactly $R-L+1$ occurrences, these are all occurrences of $q$.  Hence every occurrence of $q$ is preceded by the same length-$d$ string.  Moreover, no suffix outside $[\alpha_d,\beta_d]$ can have prefix $x$, since that would yield an additional occurrence of $q$ after deleting the first $d$ characters.  Thus $[\alpha_d,\beta_d]$ is the suffix-array interval of the resulting string.

Conversely, suppose all occurrences of $q$ are extended by the same length-$d$ string to the left, and let $x$ be the resulting string of length $m+d$.  Prepending the same string to all suffixes in $[L,R]$ preserves their lexicographic order.  Hence the suffixes starting at $i_d$ and $j_d$, obtained from the leftmost and rightmost leaves of $[L,R]$, are the leftmost and rightmost suffixes in the interval of $x$.  Therefore they have LCE at least $m+d$, and their ranks satisfy $\beta_d-\alpha_d=R-L$.  Thus $\mathsf{Valid}(L,R,m,d)$ holds.
\qed
\end{proof}

\begin{corollary}\label{cor:left-closure-search}
The set of valid values of $d$ is an interval $\{0,1,\ldots,d^*\}$.  Moreover, $\leftcl{q}$ has string depth $m+d^*$ and suffix-array interval $[\alpha_{d^*},\beta_{d^*}]$.
\end{corollary}
\begin{proof}
If all occurrences of $q$ share a common left extension of length $d$, then they also share every shorter prefix of that extension.  Hence validity is monotone.  The maximum valid $d^*$ is exactly the largest common left extension that preserves the occurrence set of $q$, which is the definition of the left-closure.
\qed
\end{proof}

By Corollary~\ref{cor:left-closure-search}, $d^*$ can be found by exponential search followed by binary search.  Since $d^*\leq n$, this requires $O(\log n)$ evaluations of $\mathsf{Valid}$. This is the common left-closure search performed after the right closure $q=\rightcl{aw}$ has been computed.  For the edge representation added for the extension $aw$ when its target is a maximal-repeat node, the algorithm stores the two lengths $m=|\rightcl{aw}|$ and $d^*$.  These values are later used, after reversing the orientation, to compute the offset and length values of the corresponding CDAWG edge label.

\begin{lemma}\label{lem:closure}
For every maximal repeat $w\in\MRep(T)$ and every left-extension character $a\in\LChar_T(w)$ such that $aw$ has at least two occurrences, the interval operations above compute the suffix-array interval and string depth of $\bothcl{aw}$.
\end{lemma}
\begin{proof}
The BWT range-distinct query enumerates exactly the subrange of occurrences of $w$ preceded by $a$, and LF-mapping transforms this subrange into $\range(aw)$.  Since $aw$ has at least two occurrences, this interval is non-singleton.  The common prefix length of all suffixes in this interval is the LCE of the two extreme suffixes, and hence the algorithm obtains $|\rightcl{aw}|$ correctly.  Corollary~\ref{cor:left-closure-search} then gives the maximum left growth preserving the same occurrence set, and therefore the suffix-array interval and string depth of $\leftcl{\rightcl{aw}}=\bothcl{aw}$.
\qed
\end{proof}

% DELETED: the subsection ``Skipping unary Weiner-link paths'' repeated the same operation as the left-closure search of Corollary~\ref{cor:left-closure-search}.
% DELETED: the algorithm does not need a separate hard/soft path-skipping case; it computes $q=\rightcl{aw}$ and then $\leftcl{q}$ by the same rule.
\subsection{Cost of one extension}

Each evaluation of $\mathsf{Valid}$ uses a constant number of accesses to $\SA$, $\ISA$, and $\LCE$.  Together with LF-mapping and range-distinct enumeration, the right- and left-closure computations for one left extension of a maximal repeat take $O(\log n)$ primitive operations, and hence $O(\log n\cdot\tGNP(T))$ time under the compressed suffix-tree interface for the working string $T$.

\begin{lemma}\label{lem:extension-cost}
Over all maximal repeats processed by Algorithm~\ref{alg:framework}, the interval computations for their left extensions take $O(\eL(T)\log n\cdot\tGNP(T))$ time and use $O(\eL(T))$ additional words for the added edge representations.
\end{lemma}
\begin{proof}
The number of added edge representations is $O(\eL(T))$ by the definition of $\eL(T)$.  For each representation whose target is a maximal-repeat node, Lemma~\ref{lem:closure} and the preceding paragraph show that the right and left closures are computed using $O(\log n)$ primitive operations.  Each primitive costs $\tGNP(T)$ time, and each such edge representation stores a constant number of keys and length values.
\qed
\end{proof}

\section{From Weiner Links to CDAWG Edges}\label{sec:edges}

This section interprets the edge representations computed in Section~\ref{sec:traversal} as CDAWG edges.  We use two classical characterizations: CDAWGs are obtained by compacting DAWGs~\cite{BlumerEtAl1985,BlumerEtAl1987}, and equivalently by merging suffix-tree loci whose rooted subtrees are isomorphic~\cite{NarisawaEtAl2017}.  Our construction uses these characterizations in opposite orientations.

Let the forward working string be $T$.  We traverse Weiner links~\cite{Weiner1973} in the suffix tree of $T$.  Under reversal, a left extension $aw$ of $w$ becomes a right extension $w^Ra$ of $w^R$.  Hence a Weiner link on the suffix-tree side of $T$ corresponds to an edge of the DAWG of $T^R$.

\begin{lemma}\label{lem:weiner-dawg}
Let $w$ be a suffix-tree locus of $T$, and suppose that the Weiner link by a character $a$ leads to the locus of $aw$.  Then, in the reversed string $T^R$, this link corresponds to the DAWG edge labeled $a$ from the node containing $w^R$ to the node containing $w^Ra$.  Moreover, hard Weiner links correspond to primary DAWG edges of $T^R$, and soft Weiner links correspond to secondary DAWG edges of $T^R$.
\end{lemma}
\begin{proof}
The first statement follows from reversal: an occurrence of $aw$ in $T$ is the reversal of an occurrence of $w^Ra$ in $T^R$, so prepending $a$ in $T$ is appending $a$ in $T^R$.  The hard/soft distinction is exactly the distinction between edges whose destinations are represented by explicit nodes and those whose destinations lie inside compacted paths, i.e., the primary/secondary distinction of Blumer et al.~\cite{BlumerEtAl1985,BlumerEtAl1987}.
\qed
\end{proof}

The compaction rule is decided on the forward suffix-tree side.  For a maximal repeat $w$ and a left-extension $aw$, Section~\ref{sec:traversal} computes
\[
  q=\rightcl{aw}, \qquad h=\leftcl{q}=\bothcl{aw}.
\]
The right closure $q$ is obtained from the Weiner-link destination $aw$ by extending it to the right while preserving its occurrence set.  Then $h=\leftcl{q}$ gives the maximal-repeat node used as the non-sink target.

\begin{figure}[tbh]
\centering
\includegraphics[width=1.0\textwidth]{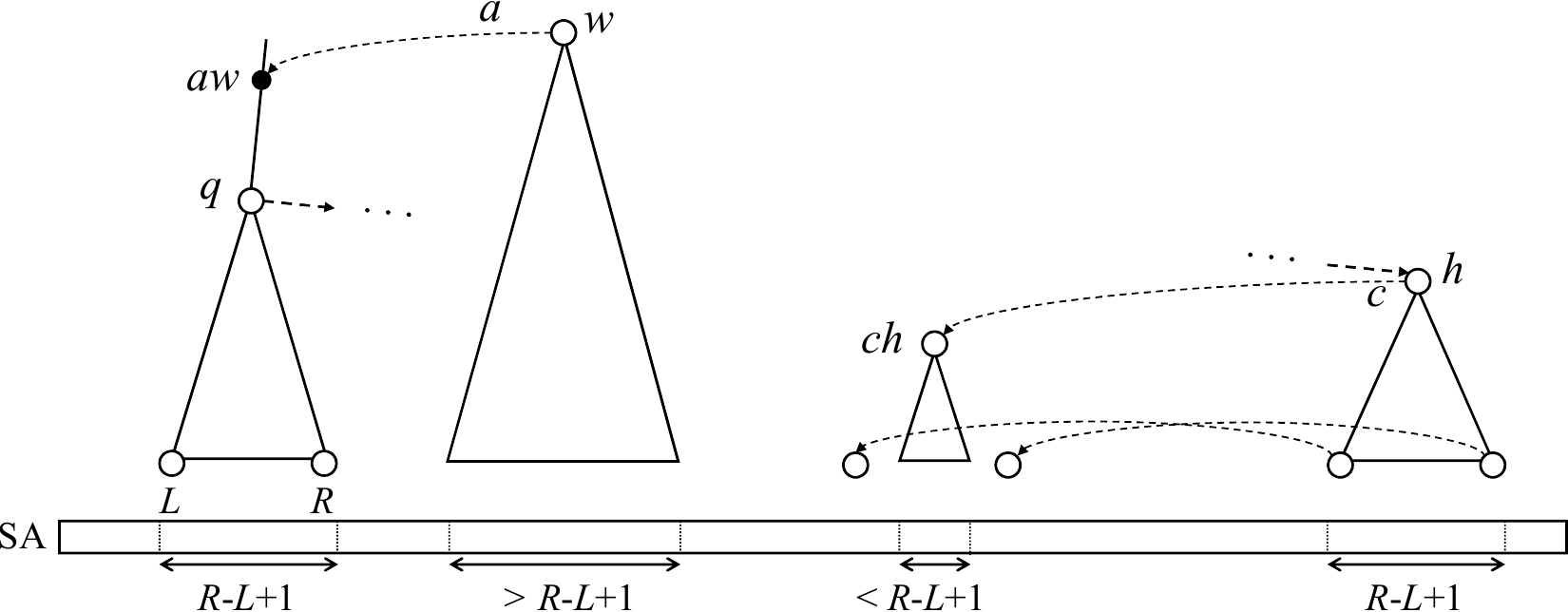}
\caption{Illustration of the interval computation for a soft Weiner link from a maximal repeat $w$.  The link by $a$ first reaches the locus of $aw$, and its right closure is $q=\rightcl{aw}$ with $\range(aw)=\range(q)=[L,R]$.  The maximal-repeat node $h=\leftcl{q}=\bothcl{aw}$ is found by the left-closure search.  Along the valid part of this search, the SA-interval width remains $R-L+1$; after the maximal left growth, the next candidate left extension $ch$, where $c$ is determined by shifting the boundary leaves one more position to the left, no longer preserves this width.  This failure is detected by the validity test on the two leaves obtained by shifting the leftmost and rightmost leaves of $[L,R]$.}
\end{figure}

\begin{lemma}\label{lem:canonical-target}
Let $w\in\MRep(T)$ and $a\in\LChar_T(w)$, and assume that $aw$ has at least two occurrences.  Let $q=\rightcl{aw}$ and $h=\bothcl{aw}$.  Then the node for $h$ is the maximal-repeat node reached from the Weiner-link destination $aw$ after merging suffix-tree loci with isomorphic rooted subtrees.  Consequently, two added edge representations whose targets are maximal-repeat nodes have the same final target node iff their canonical keys $\mathsf{key}(h)$ are equal.
\end{lemma}
\begin{proof}
The string $q=\rightcl{aw}$ is obtained by extending $aw$ to the right as long as all occurrences have the same next character.  Thus the path from the locus of $aw$ to the locus of $q$ is unary and preserves the starting positions of the occurrences.  Hence this part is precisely the non-branching part that is contracted when passing from the DAWG edge to the compact CDAWG edge.

Let $h=\leftcl{q}$, and write $h=zq$ with $d=|z|$.  By the definition of left closure, every occurrence of $q$ starts at some position $p$ iff there is an occurrence of $h$ starting at $p-d$.  Moreover, after the loci $q$ and $h$ are consumed, the remaining suffixes are identical:
\[
  T[p+|q|..]=T[(p-d)+|h|..].
\]
Therefore the rooted suffix-tree subtree below the locus of $q$ is label-isomorphic to the rooted suffix-tree subtree below the locus of $h$.  Since $q$ is right-maximal and the left closure is maximal, $h$ is both left- and right-maximal, and hence $h\in\MRep(T)$.

By the suffix-tree characterization of CDAWGs~\cite{NarisawaEtAl2017}, the loci reached from $aw$ through the unary path to $q$ and then through the above subtree-isomorphic shift to $h$ belong to the same CDAWG node.  Our construction uses the maximal-repeat locus $h$ for this CDAWG node.  Maximal repeats have unique suffix-array intervals and string depths, so equality of canonical keys is equivalent to equality of the strings $h$, and hence to equality of final target nodes.
\qed
\end{proof}

Thus hard and soft Weiner links correspond to primary and secondary DAWG edges, respectively, by Lemma~\ref{lem:weiner-dawg}.  For CDAWG construction, the algorithm computes $h=\bothcl{aw}$ for every non-unique left extension and uses $\mathsf{key}(h)$ as the target key, while unique extensions go to the sink.  The algorithm does not explicitly build either the suffix tree or the DAWG.

\subsection{Merging edge representations into CDAWG edges}

After the algorithm has added all edge representations, it sorts all canonical keys that occur as sources or non-sink targets.  Each distinct key receives one node identifier; edges whose target is the sink keep the sink as their target.  Finally, for $T=\$S\#$, reversing the orientation gives the corresponding topology for $T^R=\#S^R\$$.

\begin{lemma}\label{lem:edge-representations}
Every edge representation added by the algorithm corresponds to a left extension of a maximal repeat of $T$, and every left-oriented CDAWG edge arising from such a left extension is represented by some added edge representation.
\end{lemma}
\begin{proof}
The algorithm processes each discovered maximal repeat $w$, including $\eps$.  The BWT range-distinct query on $\range(w)$ returns exactly the characters in $\LChar_T(w)$ with their occurrence subranges.  Hence every added edge representation is generated from some pair $(w,a)$.  If $\range(aw)$ is a singleton, its target is the sink; otherwise Lemma~\ref{lem:closure} computes $h=\bothcl{aw}$ and the algorithm uses $\mathsf{key}(h)$ as the target key.  Conversely, every left-oriented CDAWG edge arising from a left extension of a maximal repeat comes from such a pair and is therefore enumerated.
\qed
\end{proof}

\begin{lemma}\label{lem:merge}
After sorting and merging equal canonical keys, the added edge representations form the left-oriented CDAWG topology of $T$.
\end{lemma}
\begin{proof}
By Lemma~\ref{lem:edge-representations}, each edge representation is added for a Weiner-link step from a maximal repeat.  Unique extensions target the sink.  For every other target, Lemma~\ref{lem:canonical-target} shows that the algorithm computes the maximal-repeat node $h=\bothcl{aw}$ reached after merging suffix-tree loci with isomorphic rooted subtrees.  Hard and soft Weiner links correspond to primary and secondary DAWG edges, respectively, but both cases use this same target-key computation.  By the suffix-tree characterization~\cite{NarisawaEtAl2017}, these maximal-repeat nodes are exactly the CDAWG nodes.  Since maximal repeats have unique suffix-array intervals and depths, two non-sink targets represent the same CDAWG node iff their canonical keys are equal.  Thus sorting the canonical keys and merging equal keys identifies exactly the non-sink CDAWG nodes, while all edges whose target is the sink keep the sink as their target.
\qed
\end{proof}

\begin{lemma}\label{lem:marker}
Let $T=\$S\#$.  Reversing the left-oriented topology of $T$ and deleting the unique edge from the node for $\eps$ whose label begins with $\#$ yields $\mathrm{CDAWG}(S^R\$)$.
\end{lemma}
\begin{proof}
Reversal transforms left extensions in $T$ into DAWG edges of $T^R=\#S^R\$$.  Hence Lemmas~\ref{lem:weiner-dawg} and~\ref{lem:merge} give the CDAWG topology of $T^R$.  The marker $\#$ occurs only at the first position of $T^R$, so every substring containing $\#$ is a unique prefix and is not a maximal repeat.  Thus the only edge containing $\#$ is the unique edge from the node for $\eps$ spelling the prefix beginning with $\#$.  Removing it leaves the CDAWG of $S^R\$$.
\qed
\end{proof}

\subsection{Compact representation of edge labels}

It remains to associate compact edge-label representations with the edges obtained after merging equal canonical keys.  We compute the information needed for label extraction in the CDAWG-grammar framework of Belazzougui and Cunial~\cite{BelazzouguiCunial2017CPM,BelazzouguiCunial2017SPIRE}, without storing edge labels explicitly.

Consider a left-oriented edge representation whose target is a maximal-repeat node and that was added for the left-extension of a maximal repeat $w_\ell$ by a character $a$.  Let
$q_\ell=\rightcl{aw_\ell}$ where $|q_\ell|=m$, and let $d$ be the additional left growth used to obtain
$h_\ell=\leftcl{q_\ell}=\bothcl{aw_\ell}$.
After reversing orientation, the corresponding right-oriented edge has source $u=w_\ell^R$ and target $v=h_\ell^R$.  
If $h_\ell=zaw_\ell y$, then $v=y^R u a z^R$.
Therefore the edge-label representation of the reversed edge $\gamma$ stores
$\gamma.\leftpos=|y|=m-|w_\ell|-1$ and 
$\gamma.\rightpos=1+d$.

\begin{lemma}\label{lem:label-representation}
For every edge representation whose target is a maximal-repeat node and that remains after merging equal canonical keys, the stored values $\gamma.\leftpos$ and $\gamma.\rightpos$ satisfy $\str(v)=x_\gamma\str(u)y_\gamma$,
where $|x_\gamma|=\gamma.\leftpos$ and $|y_\gamma|=\gamma.\rightpos$.
Consequently, the compact edge label is $y_\gamma$.
\end{lemma}
\begin{proof}
The left-oriented endpoint has the form $h_\ell=zaw_\ell y$, where the middle factor $aw_\ell$ is the raw Weiner destination, $y$ is the right growth used to obtain $q_\ell=\rightcl{aw_\ell}$, and $z$ is the additional left growth used to obtain $h_\ell=\leftcl{q_\ell}$.  Since $|q_\ell|=m$, we have $|y|=m-|w_\ell|-1$, and since $d=|z|$, reversing $h_\ell$ gives $v=y^R w_\ell^R a z^R$.  Thus $x_\gamma=y^R$ and $y_\gamma=az^R$, proving the claimed lengths and the label formula.
\qed
\end{proof}

For every maximal-repeat target node $v$, the incoming edges to $v$ are sorted by increasing $\gamma.\leftpos$.  The rank of an incoming edge is its position in this order.  Hence $\gamma.\leftpos$ can be recovered from $(v,\gamma.\rank)$ by storing an array
\[
  \mathsf{LeftByRank}_v[\gamma.\rank]=\gamma.\leftpos.
\]
Together with $\gamma.\rightpos$, this gives the information needed to recover the compact label from $\str(v)$ and $\str(u)$ without explicitly storing the label.

Edges whose target is the sink use suffix-tail edge-label representations, each requiring $O(1)$ words.

The additional level-ancestor structure used by Belazzougui and Cunial~\cite{BelazzouguiCunial2017SPIRE} can be built within the same output-size bound.  
Their structure supports prefix extraction of each edge label in linear time in the prefix length,
and is built on the spanning tree of the CDAWG with edge directions reversed.
After the edges have been ordered as above, this tree is obtained by choosing, for every node except the sink in the reversed orientation, its first outgoing edge.  Hence the tree has $O(|V|)$ edges, where $V$ is the node set of the CDAWG, and the standard level-ancestor structure is constructed in linear time and space~\cite{BenderF04}.  Since $|V|=O(\eL)$, this preprocessing is within the claimed bound.

\begin{lemma}\label{lem:labels-correct}
The edge-label representations stored by the algorithm encode exactly the compact labels of the edges of $\mathrm{CDAWG}(S^R\$)$.
\end{lemma}
\begin{proof}
For edges whose targets are maximal-repeat nodes, Lemma~\ref{lem:label-representation} shows that $\gamma.\leftpos$ and $\gamma.\rightpos$ identify the compact label $y_\gamma$.  Sorting incoming edges by $\gamma.\leftpos$ and storing $\mathsf{LeftByRank}_v$ allows $\gamma.\leftpos$ to be recovered from the target and incoming rank, as required by the Belazzougui--Cunial representation.  Edges whose target is the sink use suffix-tail edge-label representations.  The unique edge from the node for $\eps$ whose label begins with $\#$ is deleted by Lemma~\ref{lem:marker}.  Hence the remaining edge-label representations encode exactly the edge labels of $\mathrm{CDAWG}(S^R\$)$ without expanding them.
\qed
\end{proof}

\subsection{Correctness and complexity}

\begin{theorem}\label{thm:main}
Let $T=\$S\#$, let $n=|T|$, and let $\eL=\eL(T)$.  Suppose that the fully functional compressed suffix-tree interface for $T$ occupies $\sGNP(T)$ words and supports each required primitive in $\tGNP(T)$ time.  Then the algorithm constructs $\mathrm{CDAWG}(S^R\$)$ in $O(\eL\log n\cdot \tGNP(T))$ time using $O(\sGNP(T)+\eL)$ working space.
\end{theorem}
\begin{proof}
By Lemmas~\ref{lem:closure}--\ref{lem:merge}, the traversal produces exactly the edge representations of the left-oriented CDAWG topology of $T$, up to the final merging of canonical keys.  Lemma~\ref{lem:marker} converts this topology into $\mathrm{CDAWG}(S^R\$)$, and Lemma~\ref{lem:labels-correct} proves correctness of the compact edge-label representations.

Lemma~\ref{lem:extension-cost} bounds the total time for right- and left-closure computations and canonical-key computation by $O(\eL\log n\cdot\tGNP(T))$.  The remaining subtree-equivalence operations use the navigation primitives included in the black-box interface and are charged to the added edge representations.  The final sorting and merging step handles $O(\eL)$ canonical keys and takes $O(\eL\log \eL)$ time, which is bounded by $O(\eL\log n)$ and is absorbed by the preceding bound.  The compressed suffix-tree interface uses $\sGNP(T)$ words of space, while the added edge representations, keys, and edge-label representations use $O(\eL)$ additional words.
\qed
\end{proof}

\begin{corollary}\label{cor:gnp}
Let $T=\$S\#$, let $n=|T|$, let $r$ be the number of BWT-runs of $T$, and let $\eL=\eL(T)$.  Using the fully functional compressed suffix tree of Gagie, Navarro, and Prezza, the algorithm constructs $\mathrm{CDAWG}(S^R\$)$ in $O(\eL\log n\log(n/r))$ time using $O(r\log(n/r)+\eL)$ words of working space.
\end{corollary}
\begin{proof}
For the primitive set used in the proof of Theorem~\ref{thm:main}---namely $\SA$, $\ISA$, $\LCP$, $\LCE$, LF-mapping, BWT range-distinct queries, and the suffix-tree navigation needed for the subtree-equivalence test---the representation of Gagie, Navarro, and Prezza gives $\sGNP(T)=O(r\log(n/r))$ words and $\tGNP(T)=O(\log(n/r))$ time.  Substituting these bounds into Theorem~\ref{thm:main} gives the claim.
\qed
\end{proof}

\section{Concluding Remarks}

We presented an output-sensitive construction of the CDAWG of a length-$n$ string with $\eL$ edges from BWT-run bounded primitives, which runs in $O(\eL\log n\log(n/r))$ time with $O(r\log(n/r)+\eL)$ working space. A key point is that the construction goes beyond maximal-repeat enumeration: hard Weiner links discover maximal-repeat nodes, while soft Weiner links are also processed because they correspond to secondary DAWG edges under reversal. Together with existing CDAWG-based edge-label extraction and pattern-matching machinery, the constructed representation can serve as a bridge between BWT-run bounded indexes and CDAWG-based compressed string processing.

\section*{Acknowledgements}
This work was supported by JSPS KAKENHI Grant Numbers JP26K02980~(HA),
JP23K24808, JP23K18466~(SI).

\clearpage

\bibliographystyle{splncs04}
\bibliography{ref}

\end{document}